\documentclass[letterpaper, 10pt, conference]{ieeeconf}

\IEEEoverridecommandlockouts 

\author{Evagoras Makridis$^1$, Gabriele Oliva$^2$, Kasagatta Ramesh Narahari$^3$, Mohammadreza Doostmohammadian$^4$,\\Usman A. Khan$^5$, and Themistoklis Charalambous$^{1,6,7}$
\thanks{$^1$Department of Electrical and Computer Engineering, School of Engineering, University of Cyprus, Nicosia, Cyprus. E-mails: {\footnotesize\{makridis.evagoras, charalambous.themistoklis\}@ucy.ac.cy}.}
\thanks{$^2$Department of Engineering, University Campus Bio-Medico of Rome, Roma, Italy. E-mail: g.oliva@unicampus.it.}
\thanks{$^3$Nokia, Espoo, Finland. E-mail: naraharikasagattaramesh@gmail.com.}
\thanks{$^4$School of Engineering, Semnan University, Semnan, Iran. E-mail: doost@semnan.ac.ir.}
\thanks{$^5$Department of Electrical and Computer Engineering, Tufts University, Medford, Massachusetts, United States. E-mail: khan@ece.tufts.edu.}
\thanks{$^6$Department of Electrical Engineering and Automation, School of Electrical Engineering, Aalto University, Espoo, Finland.}
\thanks{$^7$FinEst Centre for Smart Cities, Tallinn, Estonia.}
\thanks{The work of E. Makridis and T. Charalambous was partly supported by the European Research Council (ERC) Consolidator Grant MINERVA (Grant agreement No. 101044629).}
}

\let\labelindent\relax

\usepackage{cite}
\usepackage{soul}
\usepackage{pifont}
\usepackage[usenames,dvipsnames]{xcolor}
\usepackage[utf8]{inputenc} 
\usepackage{url} 
\usepackage{booktabs} 
\usepackage{amsfonts,amsmath,amssymb}       
\usepackage{amsthm} 
\usepackage{nicefrac} 
\usepackage{microtype} 
\usepackage{graphicx}
\usepackage{dsfont}
\usepackage{enumitem}
\usepackage{tikz} 
\usepackage{tkz-berge}
\usepackage{pgfplots}
\usepackage{baskervillef}
\usepackage{multicol}
\usepackage{breqn}
\usepackage{tabularx}
\pgfplotsset{compat=newest}
\usetikzlibrary{arrows, chains, positioning, shapes.geometric, shapes.symbols, automata, petri, topaths,plotmarks,arrows.meta,spy,backgrounds,pgfplots.groupplots}
\usepgfplotslibrary{colorbrewer}
\usepgfplotslibrary{fillbetween}
\usepackage{hyperref}       
\hypersetup{
    colorlinks,
    linkcolor=Peach,
    citecolor=MidnightBlue,
}

\newcommand{\cmark}{\textcolor{YellowGreen}{\ding{51}}}
\newcommand{\xmark}{\textcolor{WildStrawberry}{--}}
\newcommand{\vect}[1]{\mathbf{#1}} 
\newcommand{\grvect}[1]{\boldsymbol{#1}} 
\newcommand{\set}[1]{\mathcal{#1}} 
\newcommand{\ie}{\textit{i.e.,~}} 
\newcommand{\etal}{\textit{et.~al.~}} 
\newcommand{\inneighbor}[1]{\set{N}_{#1}^{\scriptsize\texttt{in}}}
\newcommand{\inneighborS}[1]{\set{\tilde{N}}_{#1}^{\scriptsize\texttt{in}}}
\newcommand{\outneighbor}[1]{\set{N}_{#1}^{\scriptsize\texttt{out}}}
\newcommand{\indegree}[1]{d_{#1}^{\scriptsize\texttt{in}}}
\newcommand{\outdegree}[1]{d_{#1}^{\scriptsize\texttt{out}}}

\usepackage[font=footnotesize,skip=2pt]{caption}

\DeclareSymbolFont{matha}{OML}{txmi}{m}{it}
\DeclareMathSymbol{\varv}{\mathord}{matha}{118}

\newtheorem{thm}{Theorem}[]
\newtheorem{lem}{Lemma}[]

\theoremstyle{definition}

\newcommand\Label[1]{&\text{\refstepcounter{equation}(\theequation)\ltx@label{#1}}&}

\definecolor{qual1}{HTML}{7FC97F} 
\definecolor{qual2}{HTML}{BEAED4} 
\definecolor{qual3}{HTML}{FDAE61} 
\definecolor{qual4}{HTML}{4575B4} 
\definecolor{qual5}{HTML}{D73027} 

\title{\LARGE \bf Distributed Optimization with Gradient Tracking over\\Heterogeneous Delay-Prone Directed Networks}

\begin{document}
\maketitle

\begin{abstract}
In this paper, we address the distributed optimization problem over unidirectional networks with possibly time-invariant heterogeneous bounded transmission delays. In particular, we propose a modified version of the Accelerated Distributed Directed OPTimization (ADD-OPT) algorithm, herein called \emph{Robustified ADD-OPT} (R-ADD-OPT), which is able to solve the distributed optimization problem, even when the communication links suffer from heterogeneous but bounded transmission delays. We show that if the gradient step-size of the R-ADD-OPT algorithm is within a certain range, which also depends on the maximum time delay in the network, then the nodes are guaranteed to converge to the optimal solution of the distributed optimization problem. The range of the gradient step-size that guarantees convergence can be computed \emph{a~priori} based on the maximum time delay in the network.
\end{abstract}


\section{Introduction}\label{sec:introduction}

Recently, there has been an increasing number of multi-agent optimization problems where a group of agents aims at optimizing the sum of their objective functions, by allowing local information exchange over an underlying communication network. Such formulation arises in cases where  each agent has a local objective and the objective is to reach a compromise without resorting to central coordination. Such a distributed formulation is extensively used to solve multi-node optimization problems for several large-scale applications, including distributed energy resource management \cite{9366810,8825499}, resource allocation in data centers~\cite{Grammenos:2023}, large-scale machine learning \cite{tsianos2012consensus,boyd2011distributed}, and distributed localization in sensor networks \cite{khan2009diland,safavi2018distributed}.

In particular, in a multi-agent system comprised of a set of $n$ agents communicating over a network, each agent $j$, represented by a node in a graph, has access only to its own local objective function $f_{j}: \mathbb{R} \rightarrow \mathbb{R}$, and aims at minimizing the sum of objectives $\sum_{j=1}^n f_{j}(z)$, where the common decision variable\footnote{To simplify the presentation and analysis, in this work, we use a scalar decision variable. However, the results presented herein can be readily extended to vector variables.} $z\in\mathbb{R}$ is computed via information exchange and agreement with its neighbors. More specifically, the agents aim at solving cooperatively the following optimization problem:%
\begin{equation}\label{eq:optimization_problem}
    \min_{z} \frac{1}{n} \sum_{j=1}^{n} f_{j}(z).
\end{equation}

Over the past decade, considerable research efforts have been dedicated in consensus-based distributed optimization problems over networks with bidirectional links, forming undirected graphs. Such contributions have been extensively studied and presented in \cite{qu2017harnessing,yuan2016convergence,shi2015extra}. These works exploit the symmetry of bidirectional communication structure to enable the construction of doubly-stochastic consensus weight matrices, so that local decision variables of networked agents collaboratively converge towards the solution of the distributed optimization problem in \eqref{eq:optimization_problem}.

While initial efforts yielded valuable insights and algorithms, attention has shifted to more realistic scenarios with unidirectional network links, where the information might flow in only one direction. This attribute is prevalent in wireless communication due to varying transmission power and interference levels. Such (possibly asymmetric) networks are represented as directed graphs (digraphs), where directed edges represent unidirectional communication links. To address the distributed optimization problem over unidirectional links, there have been different approaches that are based on the following distributed average consensus methods: (a) the distributed weight balancing \cite{dominguez2013distributed} where agents adjust the weights on their outgoing links to asymptotically reach weights that form a doubly-stochastic matrix, (b) the surplus-based approach \cite{cai2012average} where the asymmetry of the digraph is mitigated by augmenting the state of each agent with an additional variable (often called \emph{surplus}) that locally tracks individual state updates through the assignment of weights on both the incoming and outgoing links that form both row-stochastic and column-stochastic matrices, respectively, and (c) ratio consensus \cite{dominguez2010coordination} where agents run an extra concurrent iteration that asymptotically learns the right eigenvector of the column-stochastic matrix that is formed by having agents assign weights on their outgoing links. For more information and comparison between these distributed average consensus methods over digraphs, the interested reader is referred to \cite{hadjicostis2018distributed}.

In the context of distributed optimization over reliable and time-invariant digraphs, Xi \etal introduced two algorithms based on the ratio consensus protocol: the Directed EXact firsT-ordeR Algorithm (DEXTRA) \cite{xi2017dextra} and ADD-OPT \cite{xi2017add} that incorporates \emph{gradient tracking} to record and correct the locally aggregated gradients at each agent, and accelerate the convergence to the optimal solution. Moreover, Nedich \etal in \cite{nedic2014distributed,nedic2017achieving} proposed the Subgradient-push and Push-DIGing algorithms for time-varying digraphs, based on the push-sum consensus algorithm. Using the surplus-consensus method, Xi \etal introduced the Directed-Distributed Projected Subgradient (D-DPS) and Directed-Distributed Gradient Descent (D-DGD) algorithm, respectively, to solve distributed optimization problems over digraphs, in \cite{xi2016distributed} and \cite{xi2017distributed}, respectively.


Although the aforementioned distributed optimization algorithms can successfully overcome the asymmetry of unidirectional networks and potentially find the optimal solution in \eqref{eq:optimization_problem}, none of them considered unreliable networks where the information flow between agents might be delayed. This factor arises inherently as the number of agents in the network grows, and hence a corresponding increase in the quantity of communication links can occur. This surge of links often results in information congestion within the network, giving rise to transmission delays that exert a significant degradation on the performance and efficacy of distributed optimization algorithms. Extensive work has been done, towards the development of mitigation strategies for addressing the challenges posed by congested networks, including the transmission delays encountered over communication links. These challenges have been analysed in the works presented in \cite{hadjicostis2013average,makridis2023utilizing}, in the context of distributed average consensus. Interestingly, in the context of distributed optimization, an early work in \cite{tsianos2011distributed} by Tsianos \etal analyzed a push-sum consensus-based distributed optimization algorithm for digraphs that can handle fixed bounded communication delays. However, this work lacks the analysis of convergence, as well as the incorporation of gradient tracking, a crucial feature known for expediting the convergence to the optimal solution \cite{augmented2015xu,song2023optimal,gradienttracking2023notarnicola}.

The aforementioned works illustrate the promising potential in the study and development of communication-aware strategies tailored to overcome the challenges posed by communication delays in the context of consensus-based distributed optimization algorithms. To the best of our knowledge, our work is the first to propose and analyze a distributed optimization algorithm with gradient tracking over directed graphs in the presence of time-invariant delays, as highlighted in Table~\ref{tab:distributed_opt_algorithms}.


\begin{table}[ht]
\centering
\caption{Gradient-based distributed optimization algorithms over directed graphs ($k$ denotes the number of algorithm iterations, and $0<\mu<1$ the convexity constant, introduced in the subsequent section).}
\renewcommand{\arraystretch}{1.5}
\noindent\begin{tabularx}{\columnwidth}{|X|c|c|c|X|}
    \cline{1-4}
    \hline
    \hline
    Algorithm  & Delays & Converg. Rate & Grad. Tracking\\
    \hline
    \hline
    D-DGD \cite{xi2017distributed}  & \xmark & $\mathcal{O}\left(\frac{\ln k}{\sqrt{k}}\right) $ & \xmark\\
    \hline
    D-DPS \cite{xi2016distributed} & \xmark & $\mathcal{O}\left(\frac{1}{\sqrt{k}}\right)$ & \xmark \\
    \hline
    S-Push \cite{nedic2014distributed} & \xmark & $\mathcal{O}\left(\frac{\ln k}{\sqrt{k}}\right)$ & \xmark\\
    \hline
    PS-DDA \cite{tsianos2011distributed} & fixed & $\mathcal{O}\left(\frac{1}{\sqrt{k}}\right)$ & \xmark\\
    \hline
    DEXTRA \cite{xi2017dextra} &  \xmark & $\mathcal{O}\left(\mu^k\right)$ & \cmark\\
    \hline
    ADD-OPT \cite{xi2017add} & \xmark & $\mathcal{O}\left(\mu^k\right)$ & \cmark\\
    \hline
    Push-DIGing \cite{nedic2017achieving} & \xmark & $\mathcal{O}\left(\mu^k\right)$ & \cmark \\
    \hline
    \textcolor{Black}{\textbf{R-ADD-OPT}} & fixed & $\mathcal{O}\left(\mu^k\right)$ &  \cmark\\
    \hline
    \hline
\end{tabularx}\label{tab:distributed_opt_algorithms}
\end{table}

In this work, we tackle the problem of distributed optimization over directed graphs, where the transmission links are prone to heterogeneous communication delays. Specifically, we consider how ADD-OPT algorithm \cite{xi2017add} can be modified in order to operate even in the presence of heterogeneous communication delays. Towards this end, we develop a framework that seamlessly integrates handling mechanisms of heterogeneous delays induced in the transmission links, by embedding the robustified Ratio Consensus protocol in the classical ADD-OPT algorithm. The proposed algorithm, herein called the Robustified ADD-OPT (R-ADD-OPT) algorithm, enables nodes that communicate over strongly connected digraphs, to converge to the optimal solution of the distributed optimization problem, even in the presence of heterogeneous delays. In our theoretical analysis, we demonstrate that the convergence of nodes' decision variables to the optimal solution remains assured, given that the gradient step-size is within a specific range. The range for which the algorithm is guaranteed to converge to the optimal solution is dependent on the maximum delay (provided it remains bounded) existing in the network. The proof of convergence is based on augmenting the network's corresponding weighted adjacency matrix, to handle bounded time-invariant delays. 
It is conjectured via simulations that the R-ADD-OPT algorithm works also for bounded time-varying delays, for a gradient step-size within a specific range, upper bounded by the step-size for the maximum time-invariant delay. However, the proof for this (or a tighter) range remains an open problem.

\section{Preliminaries}\label{sec:preliminaries}
\subsection{Notation}
In this paper, we denote vectors by lowercase bold letters, and matrices by uppercase italic letters. The $n \times n$ identity matrix is denoted by $I_n$, and the $n$ dimensional column vectors of all ones and zeros are represented by $\mathbf{1}_n$ and $\mathbf{0}_n$, respectively. The Kronecker product of two matrices $A$ and $B$ is denoted by $A \otimes B$. For any $f(\mathbf{x}), \nabla f(\mathbf{x})$ denotes the gradient of $f$ at $\mathbf{x}$. The spectral radius of a matrix $A$ is represented by $\rho(A)$. The right and left eigenvectors of an irreducible column-stochastic matrix $A$, corresponding to the eigenvalue of 1, are denoted by $\boldsymbol{\pi}$ and $\mathbf{1}_n^{\top}$, respectively, such that $\mathbf{1}_n^{\top} \boldsymbol{\pi}=1$. By $\|\cdot\|$, and depending on its argument, we denote a particular matrix norm, or a vector norm that is compatible with this particular matrix norm, \ie $\|A \mathbf{x}\| \leq \|A\|\|\mathbf{x}\|$ for all matrices, $A$, and all vectors, $\mathbf{x}$. The Euclidean norm of vectors and matrices is denoted by $\|\cdot\|_2$. 

\subsection{Network Model}
Consider $n$ agents (represented by graph's nodes) communicating over a \emph{strongly connected network}\footnote{A digraph $\set{G}$ is strongly connected if for every pair of vertices $v_i$ and $v_j$ $\in \set{V}$, $v_j$ is reachable by a directed walk from $v_i$}, $\set{G}=(\set{V}, \set{E})$, where $\set{V}=\{v_1, \cdots, v_n\}$ is the set of nodes and $\set{E} \subseteq \set{V} \times \set{V}$ is the set of edges (representing the communication links between agents). The total number of edges in the network is denoted by $m=|\set{E}|$. A directed edge $\varepsilon_{ji} \triangleq (v_j, v_i) \in \set{E}$, where $v_j, v_i \in \set{V}$, represent that node $v_j$ can receive information from node $v_i$, \ie $v_i \rightarrow v_j$. The nodes that transmit information to node $v_j$ directly are called in-neighbors of node $v_j$, and belong to the set $\inneighbor{j}=\{v_i \in \set{V} | \varepsilon_{ji} \in \set{E}\}$, with its cardinality denoted by $\indegree{j} = |\inneighbor{j}|$ called the in-degree. The nodes that receive information from node $v_j$ directly are called out-neighbors of node $v_j$, and belong to the set $\outneighbor{j}=\{v_l \in \set{V} | \varepsilon_{lj} \in \set{E}\}$, with its cardinality $\outdegree{j}= |\outneighbor{j}|$ called the out-degree. Note that self-loops are included in digraph $\mathcal{G}$ and this implies that the number of in-going links of node $v_j$ are ($\indegree{j} +1$) and similarly the number of its out-going links is ($\outdegree{j} +1$).

\subsection{Robustified Ratio Consensus over Directed Graphs}\label{subsec:robustified_ratio_consensus}
To reach an agreement on the common decision variable $z$ in \eqref{eq:optimization_problem} in a distributed way and over digraphs, one can employ the \emph{Ratio Consensus} protocol \cite{dominguez2011distributed} such that all nodes in the network converge to the network-wide average of their initial values. A protocol that handles time-varying (yet bounded) communication delays to ensure asymptotic average consensus was proposed in \cite{hadjicostis2013average}. Consider that node $v_{j}$ undergoes an \emph{a priori} unknown delay $\tau^{ji}$ bounded by a positive integer $\bar{\tau}^{ji}$, \ie $\tau^{ji} \leq \bar{\tau}^{ji}<\infty$. The maximum delay in the network is denoted by $\bar{\tau}=\max\{\bar{\tau}^{ji}\}$. The own value of node $v_j$ is always instantly available without delay, \ie $\tau^{jj}=0$. At each time step $k$, each node $v_j$ maintains a state variable $x_k^j \in \mathbb{R}$ (initialized at $x_0^j=V^j$, where $V^j$ is an arbitrary value of node $v_j$), an auxiliary scalar variable, $y_k^j \in \mathbb{R}^+$ (initialized at $y_0^j=1$), and $z_k^j \in \mathbb{R}$ set to $z_k^j=x_k^j/y_k^j$. Based on this notation, the strategy proposed in \cite{hadjicostis2013average} involves each node iteratively updating its states according to:
\begin{subequations}\label{eq:robustified_ratio_consensus}
    \begin{align}
        x^{j}_{k+1} &= p_{jj} x_{k}^{j} + \sum_{v_i \in \inneighbor{j}} p_{ji} x^{i}_{k-\tau^{ji}},\\
        y^{j}_{k+1} &= p_{jj} y_{k}^{j} + \sum_{v_i \in \inneighbor{j}} p_{ji} y^{i}_{k-\tau^{ji}},
    \end{align}
\end{subequations}
where $P=\{p_{ji}\} \in \mathbb{R}_{+}^{n \times n}$ forms a nonnegative column-stochastic matrix, since the weights $p_{lj}$ are assigned using:
\begin{align}\label{eq:weights}
    p_{lj}=\begin{cases}
    \frac{1}{1 + \outdegree{j}}, & v_l \in \outneighbor{j} \cup \{v_j\},\\
     0, & \text { otherwise}.
    \end{cases}
\end{align}
Since each node $v_j$ assigns the weights according to \eqref{eq:weights}, it is required that the nodes have the knowledge of their out-degree. 
Clearly, in the absence of delays, this strategy reduces to the Ratio Consensus algorithm in \cite{dominguez2011distributed}. By construction, matrix $P$ is primitive column-stochastic, and the variables $y_k^j$ are set to $1$. Then, the limit of the ratio $x_k^j$ over $y_k^j$, is the average of the initial values and is given by  \cite{kempe2003gossip, dominguez2011distributed}:
\begin{align}
\lim _{k \rightarrow \infty} z_k^j=\lim _{k \rightarrow \infty} \frac{x_k^j}{y_k^j}=\frac{\left( \sum_{i=1}^{n} x_k^i\right) \grvect{\pi}^{j}}{n\grvect{\pi}^{j}}=\frac{1}{n}\sum_{i=1}^{n} x_0^i.
\end{align}

\subsection{Assumptions}

Prior introducing the R-ADD-OPT algorithm, we make the following assumptions:   
\begin{enumerate}
[label=\textbf{Assumption \arabic*:},parsep=3pt,leftmargin=70pt,wide, labelindent=0pt]
    \item \textbf{(Strong connectivity)}\\The communication graph $\set{G}(\set{V},\set{E})$ is a strongly connected digraph. 
    \item \textbf{(Out-neighborhood knowledge)}\\Each node $v_j \in \set{V}$ knows its out-degree, \ie the number of out-neighboring nodes.
    \item \textbf{(Lipschitz-continuous gradients and strong convexity)}\\ Each local function $f_{i}$ is differentiable, strongly convex, and its gradient is globally Lipschitz-continuous, \ie for any $i$ and $z_{1}, z_{2} \in \mathbb{R}$ :\\
        a) there exists a positive constant $L$ such that
        $$
        \left\|\nabla f_{i}\left(z_{1}\right)-\nabla f_{i}\left(z_{2}\right)\right\| \leq L\left\|z_{1}-z_{2}\right\|,
        $$
        b) there exists a positive constant $\mu$ such that
        $$
        f_{i}\left(z_{1}\right)-f_{i}\left(z_{2}\right) \leq \nabla f_{i}\left(z_{1}\right)^{\top}\left(z_{1}-z_{2}\right)-\frac{\mu}{2}\left\|z_{1}-z_{2}\right\|^{2}.$$ 
        Notably strong convexity of the functions $f_i$ imply that \eqref{eq:optimization_problem} has a unique and bounded global optimal solution $z^{*}$.
    \item \textbf{(Time-invariant heterogeneous delays)}\\ The delay on link $\varepsilon^{ji}$ is denoted by $0\leq \tau^{ji} \leq \bar{\tau}^{ji} \leq \bar{\tau}$, where $\bar{\tau}^{ji}$ is the maximum delay on link $\varepsilon^{ji}$, and $\bar{\tau}$ is the maximum delay of all links in the digraph, \ie $\bar{\tau}=\max \left\{\tau^{ji}\right\}$. 
\end{enumerate}

%

\section{Accelerated Distributed Optimization over Delayed Unidirectional Networks}
In this section, we mitigate the impediments of induced communication delays, by embedding delay handling mechanisms into the ADD-OPT algorithm. However, such a strategy does not guarantee the convergence of network nodes' decision variables to the optimal solution of problem \eqref{eq:optimization_problem} when assuming a gradient step-size within the range for the original ADD-OPT algorithm. Hence, to guarantee the convergence of the R-ADD-OPT algorithm to the optimal solution, we provide a step-size analysis by computing the allowable step-size range which depends on the maximum delay in the network, given that it is bounded. 

\subsection{Robustified ADD-OPT Algorithm (R-ADD-OPT)}

We first introduce a modified version of the ADD-OPT algorithm that handles heterogeneous time-invariant, hereinafter called R-ADD-OPT. In particular, the Robustified Ratio Consensus algorithm is embedded into the ADD-OPT algorithm to ensure that the local decision variables $z^j$ reach the average consensus value $z$ regardless the transmission delays on the links, while the ADD-OPT algorithm shifts the coordinated value $z$ to the optimal solution of the optimization problem $z^*$. Similarly to the original ADD-OPT algorithm, each node $v_j \in \set{V}$ maintains three scalar variables, $x_k^j, y_k^j, w_k^j$ all $\in \mathbb{R}$, and sets $z_k^j=x_k^j/y_k^j$. At each iteration $k$, node $v_j$ assigns the weights to its states, and iteratively updates its variables at each time step $k\geq0$, according to:
\begin{align}\label{eq:r-add-opt}
x^{j}_{k+1} \! &= \!\!\! \sum_{v_i \in \inneighborS{j}} (p_{ji} x^{i}_{k-\tau^{ji}} - \alpha w^{i}_{k-\tau^{ji}}) \!-\! \alpha w^{j}_{k},\nonumber\\
y^{j}_{k+1} \! &= \!\!\! \sum_{v_i \in \inneighborS{j}}  p_{ji} y^{i}_{k-\tau^{ji}},\nonumber\\
z_{k+1}^{j} \! &=\frac{x_{k+1}^{j}}{y_{k+1}^{j}},\nonumber\\
w_{k+1}^{j} \! &= \!\!\! \sum_{v_i \in \inneighborS{j}}  p_{ji} w_{k-\tau^{ji}}^{i} \!+\! \nabla f_{j}\left(z_{k+1}^{j}\right) \!-\! \nabla f_{j}\left(z_{k}^{j}\right),
\end{align}
where $\inneighborS{j}=\{\inneighbor{j}\cup\{v_j\}\}$, and each node is initialized with $y_0^j=1$, $w_0^j=\nabla f_j(z_0^j)$, and arbitrary scalar values for $x_0^j$, and $z_0^j$. Clearly, when $\bar{\tau}=0$, then the R-ADD-OPT algorithms reduces to the ADD-OPT algorithm in \cite{xi2017add}. 


\subsection{Delayed Digraph Augmentation}
To model possibly delayed information exchange between the network nodes, we employ an augmented graph representation (as in \cite{hadjicostis2013average}) by adding extra virtual nodes that represent the delays on the links. The number of extra virtual nodes representing the delays on the links is bounded by the maximum delay in the network, $\bar{\tau}$. Hence, for each node $v_j \in \set{V}$, we add $\bar{\tau}$ extra virtual nodes, $v_{j}^{(1)}, v_{j}^{(2)}, \ldots, v_{j}^{(\bar{\tau})}$, where the virtual node $v_{j}^{(r)}$ holds the information that is destined to arrive at node $v_j$ after $r$ time steps. Thus, the augmented digraph consists of at most $n(\bar{\tau}+1)$ nodes and $(1+2 \bar{\tau})|\mathcal{E}|$ edges. In general, we augment a network of $n=|\mathcal{V}|$ nodes, by introducing $n\bar{\tau}$ nodes which results to a total of $\bar{n}=n(\bar{\tau}+1)$ nodes. An example of a two-node digraph where the information transmitted over the link $\varepsilon^{21}$ is delayed by $2$ time steps, and for the link $\varepsilon^{12}$ is delayed by $1$ time step, is provided in Fig.\ref{fig:delayed_graph}.
\begin{figure}[h]
    \centering
    \includegraphics{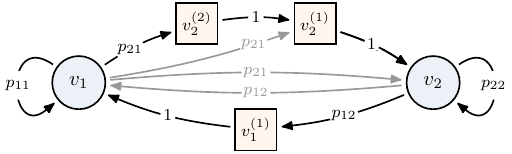}
    \caption{Example of the augmented digraph with fixed transmission delays.}
    \label{fig:delayed_graph}
\end{figure}

Based on the graph augmentation, we extend the original vectors of the ADD-OPT algorithm, to keep the information for all nodes and all possible delays, as follows: $\setlength{\arraycolsep}{0pt}\vect{\hat{x}}_{k}=
\left[\begin{array}{c}
\vect{x}_{k,(0)},\ldots, \vect{x}_{k,(\bar{\tau})}
\end{array}\right]^{\top}$,\;\;%
$\setlength{\arraycolsep}{0pt}\vect{\hat{y}}_{k}=
\left[\begin{array}{c}
\vect{y}_{k,(0)},\ldots, \vect{y}_{k,(\bar{\tau})}
\end{array}\right]^{\top}$,\;\;%
$\setlength{\arraycolsep}{0pt}\vect{\hat{z}}_{k}=
\left[\begin{array}{c}
\vect{z}_{k,(0)},\ldots, \vect{z}_{k,(\bar{\tau})}
\end{array}\right]^{\top}$,\;\;%
$\setlength{\arraycolsep}{0pt}\vect{\hat{w}}_{k}=
\left[\begin{array}{c}
\vect{w}_{k,(0)},\ldots, \vect{w}_{k,(\bar{\tau})}
\end{array}\right]^{\top}$, \; and %
$\setlength{\arraycolsep}{0pt}\nabla \vect{\hat{f}}(\vect{\hat{z}}_k)=
\left[\begin{array}{c}
\nabla \vect{f}(\vect{z}_{k,(0)})~\mathbf{0}_{n\bar{\tau}}^{\top})
\end{array}\right]^{\top}$, where $\nabla\vect{f}(\vect{z}_{k,(0)}) = \left[ \nabla f_1(z_1(k)), \ldots, \nabla f_n(z_n(k))\right]^{\top}$. Note that, the gradients of the virtual nodes $\nabla\vect{f}(\vect{z}_{k,(r)})$, for $r\neq 0$, are always zero.
The vector $\vect{x}_{k,(0)}$ contains all scalar states of all actual nodes, while $\vect{x}_{k,(r)}$ contains all scalar states of the virtual nodes generated for delay $r$. Hence, the R-ADD-OPT algorithm can be written in matrix-vector form (for the analysis presented in the subsequent sections) as follows:
\begin{subequations}\label{eq:r-add-opt-matrix-form}
\begin{align}
\mathbf{\hat{x}}_{k+1} &=\Xi \mathbf{\hat{x}}_{k}- \alpha \mathbf{\hat{w}}_{k}\label{eq:x}, \\
\mathbf{\hat{y}}_{k+1} &=\Xi \mathbf{\hat{y}}_{k}\label{eq:y}, \\
\mathbf{\hat{z}}_{k+1} &=Y^{-1}_{k+1} \mathbf{\hat{x}}_{k+1},\label{eq:z}\\
\mathbf{\hat{w}}_{k+1} &=\Xi \mathbf{\hat{w}}_{k}+\nabla \mathbf{\hat{f}}(\vect{\hat{z}}_{k+1})-\nabla \mathbf{\hat{f}}(\vect{\hat{z}}_k),\label{eq:w}
\end{align}
\end{subequations}
where $Y_k=\operatorname{diag}\left(\mathbf{\hat{y}}_k\right)$. Additionally, we have the initial conditions, $\vect{\hat{x}}_{0} = [V^{\top}~\mathbf{0}_{n\bar{\tau}}^{\top}]^{\top}$, $\hat{\mathbf{y}}_{0}= [\mathbf{1}_{n}^{\top} ~\mathbf{0}_{n\bar{\tau}}^{\top}]^{\top}$\footnote{Note, that the ratio $\vect{z}_{k,(1)},\ldots,\vect{z}_{k,(\bar{\tau})}$ of the virtual nodes could involve a division by zero since their corresponding $\vect{y}_{k,(1)},\ldots,\vect{y}_{k,(\bar{\tau})}$ initiates at 0. However, the values of the virtual nodes are never used to evaluate gradients but they are introduced for analysis purposes instead.}, $\hat{\mathbf{w}}_{0}=\nabla \hat{\mathbf{f}}_{0}(\vect{\hat{z}}_0)$. Matrix $\Xi_k \in \mathbb{R}_{+}^{\bar{n} \times \bar{n}}$ is a nonnegative matrix associated with the augmented digraph:
\begin{align}\label{eq:Xi}
\Xi \triangleq\left(\begin{array}{ccccc}
P^{(0)} & I_{n \times n} & 0 & \cdots & 0 \\
P^{(1)} & 0 & I_{n \times n} & \cdots & 0 \\
\vdots & \vdots & \vdots & \ddots & \vdots \\
P^{(\bar{\tau}-1)} & 0 & 0 & \cdots & I_{n \times n} \\
P^{(\bar{\tau})} & 0 & 0 & \cdots & 0
\end{array}\right)
\end{align} 
where each element of $P^{(r)}$ is determined by:
\begin{align}\label{eq:weighted_delayed_link}
P^{(r)}(j,i)= \begin{cases}P(j,i), & \text { if } \tau^{ji}=r, \; (j,i) \in \mathcal{E}, \\ 0, & \text { otherwise. }\end{cases}    
\end{align}
Clearly, matrix $\Xi$ maintains column-stochasticity, although the information sent over the links might be delayed.

\section{Convergence Analysis}\label{sec:main_results}
To analyze the R-ADD-OPT in the presence of transmission delays, we further define the following notation:%
\begin{subequations}\label{eq:generic_notation}
\begin{align*}
  \bar{\mathbf{x}}_{k} &=\frac{1}{\bar{n}} \mathbf{1}_{\bar{n}} \mathbf{1}_{\bar{n}}^{\top} \mathbf{\hat{x}}_{k} \Label{eq:x_matrix_form} & \vect{\hat{z}}^{*} &= [\mathbf{1}_{n}^{\top}z^{*}~ \mathbf{0}_{\bar{n}-n}^{\top}]^{\top} \Label{eq:z_matrix_form}\\
  \bar{\mathbf{w}}_{k} &=\frac{1}{\bar{n}} \mathbf{1}_{\bar{n}} \mathbf{1}_{\bar{n}}^{\top} \mathbf{\hat{w}}_{k}\Label{eq:w_matrix_form} 
  & \xi &=\left\|\Xi-I_{\bar{n}}\right\|_2 \Label{eq:xi_matrix_form}\\
  \vect{\bar{g}}^{\hat{z}}_{k} &=\frac{1}{\bar{n}} \mathbf{1}_{\bar{n}} \mathbf{1}_{\bar{n}}^{\top} \nabla \mathbf{\hat{f}}(\vect{\hat{z}}_{k})\Label{eq:g_matrix_form} 
  & \epsilon &=\left\|I_{\bar{n}}-\Xi_{\infty}\right\|_2 \Label{eq:epsilon_matrix_form}\\
  \vect{\bar{g}}^{\bar{x}}_{k} &=\frac{1}{\bar{n}} \mathbf{1}_{\bar{n}} \mathbf{1}_{\bar{n}}^{\top} \nabla \mathbf{\hat{f}}(\vect{\bar{x}}_{k})\Label{eq:h_matrix_form} & \eta &=\max (\chi,\zeta)\Label{eq:eta}
\end{align*}
\end{subequations}
where $z^*$ comes from assumption A3, $\zeta=|1-\bar{n} \alpha L|$, $\chi=|1-\bar{n} \alpha \mu|$,
$\Xi_{\infty}=\lim_{k\rightarrow \infty}\Xi^k$, and $Y_{\infty}=\lim _{k \rightarrow \infty} Y_{k}$. The convergence of $\Xi$ and $Y_{\infty}$ is shown through the Lemma~\ref{lem:sigma} and Lemma~\ref{lem:rgamma1}.

\begin{lem}[Xi \etal \cite{xi2017add}]\label{lem:sigma} 
Let assumption A1 and A2 hold. Consider $Y_{\infty}=\lim_{k\rightarrow \infty} Y_k$ and $\Xi_k=\Xi$ being the column-stochastic matrix as defined in \eqref{eq:Xi}, with its non-$\mathbf{1}_{\bar{n}}$ Perron vectors denoted by $\boldsymbol{\pi}$. Then, for any vector $\mathbf{a} \in \mathbb{R}^{\bar{n}}$, and for $\bar{\mathbf{a}}=\frac{1}{\bar{n}}\mathbf{1}_{\bar{n}} \mathbf{1}_{\bar{n}}^{\top} \mathbf{a}$, there exists $0<\sigma<1$ such that $\forall k$
\begin{align}
\left\|\Xi \mathbf{a}-Y_{\infty} \bar{\mathbf{a}} \right\|_{\boldsymbol{\pi}} \leq \sigma\left\|\mathbf{a}-Y_{\infty} \bar{\mathbf{a}} \right\|_{\boldsymbol{\pi}}.
\end{align}
\end{lem}

\begin{lem}[Nedic \etal \cite{nedic2014distributed}]\label{lem:rgamma1}
Let Assumptions A1 and A2 hold, and consider $Y_{k}$ and its limit $Y_{\infty}$ as generated from the weight matrix $\Xi$. Then, there exist $0<\gamma_1<1$ (the contraction factor defined in Lemma~\ref{lem:sigma} and $0<\psi<\infty$ such that $\forall k$
\begin{align}
\left\|Y_{k}-Y_{\infty}\right\| \leq \psi \gamma_1^{k} .
\end{align}
\end{lem}


Given the notation for the delayed case in \eqref{eq:generic_notation}, we can further denote
$\mathbf{t}_{k}, \mathbf{s}_{k} \in \mathbb{R}^{3}$, and $G, H_{k} \in \mathbb{R}^{3 \times 3}$, for all $k$ as:

\begin{align}
\mathbf{t}_{k} &=\left[\begin{array}{cc}
\left\|\vect{\hat{x}}_{k}-Y_{\infty} \vect{\bar{x}}_{k}\right\| \\
\left\|\vect{\bar{x}}_{k}-\vect{\hat{z}}^{*}\right\|_2\nonumber \\
\left\|\vect{\hat{w}}_{k}-Y_{\infty} \vect{\bar{g}}^{\hat{z}}_k\right\|
\end{array}\right], \quad \mathbf{s}_{k}=\left[\begin{array}{c}
\left\|\mathbf{\hat{x}}_{k}\right\|_2 \\
0 \\
0
\end{array}\right], \nonumber\\
G &=\left[\begin{array}{ccc}
\sigma & 0 & \alpha \\
\alpha c L \tilde{y} & \eta & 0 \\
c d \epsilon L \tilde{y}\left(\xi+\alpha L y \tilde{y}\right) & \alpha d \epsilon L^{2} y \tilde{y} & \sigma+\alpha c d \epsilon L \tilde{y}
\end{array}\right],\nonumber\\
H_{k} &=\left[\begin{array}{ccc}
0 & 0 & 0 \\
\alpha L \tilde{y} \psi \gamma_{1}^{k-1} & 0 & 0 \\
(\alpha L y+2) d \epsilon L \tilde{y}^{2} \psi \gamma_{1}^{k-1} & 0 & 0
\end{array}\right],\label{eq:matrices}&&
\end{align}
where $\sigma$ is given in Lemma~\ref{lem:xin2019}, $c$ and $d$ are positive constants from the equivalence of $\left\| \cdot \right\|$ and $\left\| \cdot \right\|_2$, $L$ is the Lipschitz-continuity constant from assumption A3, $y=\sup _{k}\left\|Y_{k}\right\|_{2}$, $\tilde{y}=\sup _{k}\left\|Y_{k}^{-1}\right\|_{2}$, and $\xi$, $\epsilon$, $\eta$ are defined in \eqref{eq:generic_notation}.

\begin{lem}[Xi \etal \cite{xin2019distributed}]\label{lem:xin2019}
Let assumption A1 hold and consider an irreducible and primitive nonnegative matrix $W \in \mathbb{R}_{+}^{n \times n}$ and its non-$\mathbf{1}_{n}$ Perron vector denoted by $\vect{v}$. Then $\forall \mathbf{a} \in \mathbb{R}^{n}$ we have: 
\begin{align}
\left\|W \mathbf{a}-W_{\infty} \mathbf{a}\right\|_{\vect{v}} \leq \sigma_{W}\left\|\mathbf{a}-W_{\infty} \mathbf{a}\right\|_{\vect{v}}
\end{align}
where $\sigma_{W} \triangleq\left\|W-W_{\infty}\right\|_{\vect{v}}<1$. One can also verify that
\begin{align}
\sigma_{W}=\sigma_2\left(\operatorname{diag}\left(\sqrt{\vect{v}}\right)^{-1} W \operatorname{diag}\left(\sqrt{\vect{v}}\right)\right)
\end{align}
where $\sigma_2(\cdot)$ is the second largest singular value of a matrix.
\end{lem}


In order to prove convergence of the R-ADD-OPT algorithm, we first need to show that the evolution of $\left\|\vect{\hat{x}}_{k}-Y_{\infty} \vect{\hat{x}}_{k}\right\|$, $\left\|\vect{\hat{x}}_k - \vect{\hat{z}}^{*} \right\|_2$, and $\left\|\vect{\hat{w}}_k-Y_{\infty} \vect{\bar{g}}^{\hat{z}}_k\right\|$ are bounded. In other words, these elements should be bounded with respect to their predecessor state. Given the notation in \eqref{eq:matrices} we introduce Lemma~\ref{lem:Xi_linear_relation}, by which we show that these elements are kept bounded (the proof is reported in the Appendix).

\begin{lem}\label{lem:Xi_linear_relation}
Consider assumptions A1-A4 for the case of time-invariant delays, and let $\mathbf{t}_{k}, \mathbf{s}_{k}, G$, and $H_{k}$ be defined as in \eqref{eq:matrices}, then the following linear relation is satisfied:%
\begin{align}\label{eq:linear_relation}
\mathbf{t}_{k} \leq G \mathbf{t}_{k-1}+H_{k-1} \mathbf{s}_{k-1} .
\end{align}
\end{lem}

Next, we show that given an appropriate selection of the gradient step-size, the R-ADD-OPT algorithm is guaranteed to converge to the optimal solution of the problem in \eqref{eq:optimization_problem}, regardless the length of time-invariant delays on the transmission links. It is worth mentioning that, the range of the gradient step-size that can guarantee convergence, can be computed \emph{a~priori} based on the maximum fixed transmission delay in the network, given that it is bounded. The step-size range that guarantees the convergence of the R-ADD-OPT algorithm to the optimal solution, is given by Lemma~\ref{lem:step-size-range}, and its proof is available in the Appendix.

\begin{lem}[Gradient step-size for convergence, Xi \etal \cite{xi2017add}]\label{lem:step-size-range}
Consider the matrix $G_{\alpha}$ being a function of the step-size $\alpha$, as defined in \eqref{eq:matrices}. It is worth mentioning that every element in $G_{\alpha}$, other than the step-size $\alpha$, is constant. Then, it follows that $\rho\left(G_{\alpha}\right)<1$ if the step-size $\alpha \in(0, \bar{\alpha})$, where
\begin{align}\label{eq:alpha}
\bar{\alpha}=\min \left\{\frac{\sqrt{\delta^{2}+4 \bar{n} \mu(1-\sigma)^{2} \theta}-\delta}{2 \theta}, \frac{1}{\bar{n} L}\right\},
\end{align}
and $\delta=\bar{n} \mu c d \epsilon L \tilde{y}(1-\sigma+\xi)$, $\theta=c d \epsilon L^2 y \tilde{y}^2 (L+\bar{n} \mu)$, while $c$ and $d$ are the constants from the equivalence of $\|\cdot\|$ defined in Lemma~\ref{lem:sigma} and $\|\cdot\|_2$.
\end{lem}

\begin{thm}
Consider the R-ADD-OPT algorithm in \eqref{eq:r-add-opt} and let assumptions A1-A4 hold. The R-ADD-OPT algorithm converges exponentially for an appropriate gradient step-size $\alpha \in (0,\bar{\alpha})$, where $\bar{\alpha}$ is given in \eqref{eq:alpha}, which depends on the upper bound of the delays in the network.
\end{thm}
\begin{proof}
In order to prove the linear convergence of R-ADD-OPT, it is sufficient to show that $\|\mathbf t_k \|$ in \eqref{eq:linear_relation} goes to zero exponentially. To this aim, it is sufficient to show that:
\begin{itemize}
    \item The spectral radius of matrix $G$ is less than 1, \ie $\rho(G)<1$, that is the largest absolute value of the eigenvalues of $G$ is less than 1. Following the results of Lemma~\ref{lem:step-size-range}, we can guarantee that $\rho(G)<1$ if the step-size is within the range $\alpha \in (0,\bar{\alpha})$, where $\bar{\alpha}$ is given in \eqref{eq:alpha}.
    \item The term $H_k$ decays in linear fashion. This can be shown since $0<\gamma_1<1$, and hence according to \cite[Lemma~5]{xi2017add}, $H_k$ decays with $k$. 
\end{itemize}
\end{proof}

\section{Numerical Example}
Consider a network of $5$ agents (nodes), represented by the digraph shown in digraph in Fig.~\ref{fig:numerical_example_graph}. In this example, the main goal of the nodes is to allocate their resources in order to minimize a global cost function of the form of \eqref{eq:optimization_problem}, where each node is endowed with a scalar local cost function $f_i(x): \mathbb{R} \rightarrow \mathbb{R}$. Such scalar local cost functions, within the context of the distributed resource allocation (DRA) problem, are often quadratic of the form:
\begin{align}\label{eq:local_cost}
\centering
f_i(x)=\frac{1}{2} \beta_i\left(x-\varphi_i\right)^2
\end{align}
where $\beta_i>0, \varphi_i \in \mathbb{R}$ is the demand in node $v_i$ (and in our case is a positive real number), and $x$ is a global optimization parameter that will determine the resource allocation at each node. The optimal allocation for problem \eqref{eq:optimization_problem} with local quadratic cost functions as in \eqref{eq:local_cost} can be solved in closed form, and its minimizer $x^*$ is given by
\begin{align}\label{eq:minimizer}
\centering
x^*=\arg \min _{x \in \mathcal{X}} \sum_{v_i \in \mathcal{V}} f_i(x) =\frac{\sum_{v_i \in \mathcal{V}} \beta_i \varphi_i}{\sum_{v_i \in \mathcal{V}} \beta_i}
\end{align}
where $\mathcal{X}$ is the set of feasible values of $x$. Note that if $\beta_i=1$ for all $v_i \in \mathcal{V}$, the optimal solution is the average consensus value.

Each node $v_j$ chooses the weights of its out-going links, as defined in \eqref{eq:weights}, and executes the R-ADD-OPT iterations in \eqref{eq:r-add-opt}. The nodes initiate their iterations at $\vect{x}_0=[4~1~5~2~3]^{\top}$, $\vect{y}_0=[1~1~1~1~1]^{\top}$, and $\boldsymbol{\beta}_0=[1~5~3~4~1]^{\top}$. 
The maximum delay in the network is $\bar{\tau} = \max\{\bar{\tau}^{ji}\}$. In this setting, the range on the step-size that guarantees convergence can be computed by setting $c=d=L=1$, $y=1.67$, $\tilde{y}=3$, $\mu=0.1$, $\epsilon=1.1$, and $\xi=1.13$ to \eqref{eq:alpha}.
\vspace{-15pt}
\begin{figure}[ht]
\resizebox{\columnwidth}{!}{\includegraphics{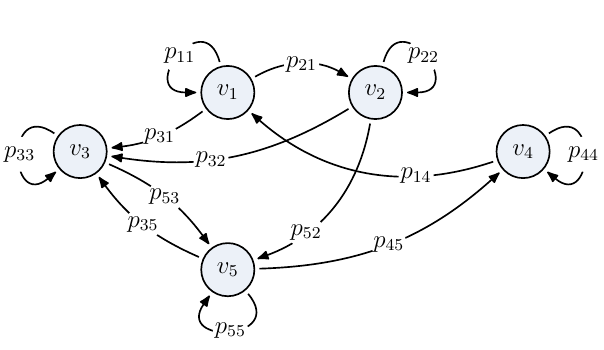}}
\caption{Digraph consisting of five agents.}
\label{fig:numerical_example_graph}
\end{figure}

To emphasize the influence of delays on the convergence speed of the R-ADD-OPT algorithm, we compute $\sigma$ as given in Lemma~\ref{lem:xin2019} for different upper bound on the delays. As shown in Table~\ref{tab:sigma_vs_delays}, with longer delays on the links, $\sigma$ approaches one as the second largest eigenvalue of $\Xi$ approaches one, while for the non-delayed case, the $\sigma$ parameter is the second largest eigenvalue of $P$. 


\begin{table}[ht]
\centering
\caption{Parameter $\sigma$ for different upper bound on the delays.}\vspace{3pt}
\noindent \renewcommand{\arraystretch}{1}
\begin{tabular}{c||c|c|c|c}  
    \cline{1-4}
    \hline
    \hline
    $\bar{\tau}$ & 0 & 2 & 5 & 10 \\
    \hline
    $\sigma$ &  0.599 & 0.877 & 0.963 &  0.987\\
    \hline
    \hline
\end{tabular}\label{tab:sigma_vs_delays}
\end{table}

Fig.~\ref{fig:spectral_radius_alpha} depicts the variation of spectral radius of matrix $G_a$, \ie $\rho(G_a)$, with respect to step-size, for different lengths of maximum delay in the network. In particular, Fig.~\ref{fig:spectral_radius_alpha} indicates the validity of the theoretical bounds on the step-sizes. Additionally, we can observe that we can achieve the best convergence rate be selecting the step-size that minimizes the spectral radius. Intuitively we can infer that for larger delays the gradient step-size should be smaller to guarantee convergence to the optimal value.

\begin{figure}[!h]
\resizebox{\columnwidth}{!}{\includegraphics{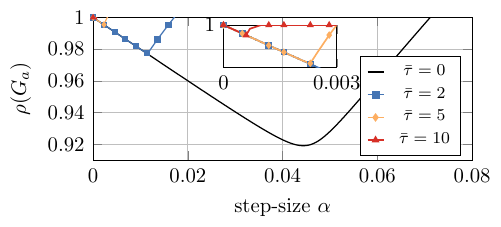}}
\caption{Spectral radius with alpha step-size value for different delays.}
\vspace{-5pt}
\label{fig:spectral_radius_alpha}
\end{figure}




To get intuition about the convergence rate of the R-ADD-OPT algorithm with respect to different lengths of delays, we plot the residual $\frac{1}{n} \sum_{i=1}^n\left(z_k^i-x^*\right)^2$, for different upper bounds on the delays. In this example we set the step-size on the value that minimizes the spectral radius as shown in Fig.~\ref{fig:spectral_radius_alpha}, and we examine both time-invariant and time-varying delays, assuming that all available (and non-self loops) links are delay-prone. In particular, for time-invariant delays, the information is delayed by exactly $\bar{\tau}$ on all delay-prone links, while for time-varying delays, the delay length is random and within $0\leq \tau_{k}^{ji} \leq \bar{\tau}^{\star}$. As shown in Fig.~\ref{fig:residual_delays_fixed_delays}, the convergence rate of the R-ADD-OPT algorithm to the optimal solution depends on the length of the delays, and the step-size should be set at lower values for longer delays on links such that the convergence is guaranteed (see Fig.~\ref{fig:spectral_radius_alpha}). Despite showing a single ensemble for the time-varying delays in Fig.~\ref{fig:residual_delays_fixed_delays}, we ran simulations for several different realizations for random integer variable $\tau_{k}^{ji}$, satisfying $0\leq \tau_{k}^{ji} \leq \bar{\tau}^{\star}$. It is conjectured that the R-ADD-OPT algorithm works also for bounded time-varying delays, as long as the step-size is within a range specified by the delay upper bound $\bar{\tau}^{\star}$. However, the proof for this range, while it is possibly a conservative bound, remains an open problem.


\begin{figure}[!h]
\resizebox{\columnwidth}{!}{
\includegraphics{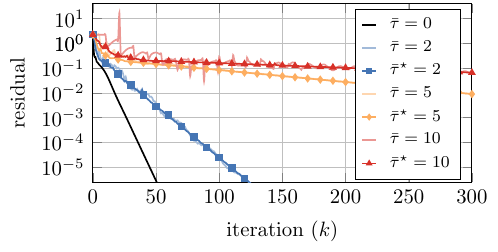}}
\caption{Residual for different upper bound on the delays in the network $\bar{\tau}$. Time-varying delays are within the bound $0\leq \tau_{k}^{ji} \leq \bar{\tau}^{\star}$.}
\label{fig:residual_delays_fixed_delays}
\end{figure}

\section{Conclusions}
In this paper, 
we proposed a modified version of the ADD-OPT algorithm, called R-ADD-OPT, that handles heterogeneous transmission delays over unidirectional networks. We showed analytically that if the gradient step-size of the R-ADD-OPT algorithm is within a certain predefined range, the agents are guaranteed to converge to the optimal solution of the distributed optimization problem, regardless the length of delays on the transmission links. The range of the step-size that can guarantee convergence, can be computed \emph{a~priori} based on the maximum transmission delay in the network, given that it is bounded. Finally, we showed with experimental simulations the potential of the R-ADD-OPT algorithm to converge to the global optimal value in the presence of time-varying delays on the communication links.

\clearpage
\section*{Appendix}
\setlength{\abovedisplayskip}{5pt}
\setlength{\belowdisplayskip}{5pt}


\begin{proof}[\textbf{Proof of Lemma~\ref{lem:sigma}}]
First, from Perron-Frobenius Theorem we know that $\rho(\Xi)=1$ since $\Xi$ is column-stochastic and irreducible. Every eigenvalue of $\Xi$ other than the largest in magnitude, \ie $|\lambda_1|$, is strictly less than $\rho(\Xi)$. Moreover, since $\Xi$ is primitive, we know that $\boldsymbol{\pi}$ is a strictly positive (right) eigenvector corresponding to the eigenvalue of 1 such that $\mathbf{1}_{\bar{n}}^{\top} \boldsymbol{\pi}=1$, and hence $\Xi_{\infty} = \lim _{k \rightarrow \infty} \Xi^k=\boldsymbol{\pi} \mathbf{1}_{\bar{n}}^{\top}$. One can also verify that $\Xi \Xi_{\infty} = \Xi_{\infty} \Xi_{\infty} = \Xi_{\infty}$. Then it follows that:%
\begin{align}
\Xi \mathbf{a} - Y_{\infty} \bar{\mathbf{a}} &= \Xi \mathbf{a} - \Xi_{\infty} \mathbf{a}\nonumber\\
&= \Xi \mathbf{a} - \Xi_{\infty} \mathbf{a} + \Xi_{\infty} \Xi_{\infty} Y_{\infty} \bar{\mathbf{a}} -  \Xi \Xi_{\infty} Y_{\infty} \bar{\mathbf{a}}\nonumber\\
&= (\Xi - \Xi_{\infty})(\mathbf{a} - Y_{\infty} \bar{\mathbf{a}}).\nonumber
\end{align}%
Therefore, considering the right eigenvector $\boldsymbol{\pi}$ of $\Xi$, we get:
\begin{align}
\left\| \Xi \mathbf{a} - Y_{\infty} \bar{\mathbf{a}}  \right\|_{\boldsymbol{\pi}} &= \left\| (\Xi - \Xi_{\infty})(\mathbf{a} - Y_{\infty} \bar{\mathbf{a}}) \right\|_{\boldsymbol{\pi}} \nonumber\\
&\overset{\tiny{\mathrm{(i)}}}{\leq} \left\| \Xi - \Xi_{\infty} \right\|_{\boldsymbol{\pi}} \left\| \mathbf{a} - Y_{\infty} \bar{\mathbf{a}} \right\|_{\boldsymbol{\pi}} \nonumber \\
&\overset{\tiny{\mathrm{(ii)}}}{\leq} \sigma \left\| \mathbf{a} - Y_{\infty} \bar{\mathbf{a}} \right\|_{\boldsymbol{\pi}},\nonumber
\end{align}
where $(i)$ stems from the sub-multiplicative matrix norm property and $(ii)$ stems from Lemma~\ref{lem:xin2019}.
\end{proof}

\begin{proof}[\textbf{Proof of Lemma~\ref{lem:Xi_linear_relation}}]
The proof follows the results of \cite{xi2017add} \emph{mutatis mutandis} by considering the augmented matrix-vector form of the R-ADD-OPT algorithm in \eqref{eq:r-add-opt-matrix-form}, which we include here for completeness. Let assumptions A1-A4 hold.
For the first element $\left\|\vect{\hat{x}}_{k}-Y_{\infty} \vect{\hat{x}}_{k}\right\|$, we first substitute $\vect{\hat{x}}_k$ from \eqref{eq:x}, and $\vect{\bar{x}}_k$ from Lemma~\ref{lem:dynamics_initial_conditions}, and we get:
\begin{align}
\left\|\vect{\hat{x}}_{k}-Y_{\infty} \vect{\bar{x}}_{k}\right\| &= \left\| \Xi \vect{\hat{x}}_{k-1} - \alpha \vect{\hat{w}}_{k-1} - Y_{\infty} (\vect{\bar{x}}_{k-1}- \alpha \vect{\bar{g}}^{\hat{z}}_{k-1})\right\| \nonumber\\
&\leq \left\|\Xi \vect{\hat{x}}_{k-1}-Y_{\infty} \vect{\bar{x}}_{k-1}\right\| \nonumber \\ 
&\quad + \alpha\left\|\vect{\hat{w}}_{k-1}-Y_{\infty} \vect{\bar{g}}^{\hat{z}}_{k-1}\right\|\nonumber\\
&\leq  \sigma\left\|\vect{\hat{x}}_{k-1}-Y_{\infty} \vect{\bar{x}}_{k-1}\right\| \nonumber \\ & \quad  +\alpha\left\|\vect{\hat{w}}_{k-1}-Y_{\infty} \vect{\bar{g}}^{\hat{z}}_{k-1}\right\|,
\end{align}%
where the last inequality holds by Lemma~\ref{lem:sigma}.

For the second element, $\left\|\vect{\bar{x}}_k - \vect{\hat{z}}^{*} \right\|_2$, we first substitute $\vect{\bar{x}}_k$ from Lemma~\ref{lem:dynamics_initial_conditions}, and we get: 
\begin{align}\label{eq:bound_step2}
    \left\|\vect{\bar{x}}_k - \vect{\hat{z}}^{*} \right\| &= \left\|\vect{\bar{x}}_{k-1} - \alpha \vect{\bar{g}}^{\hat{z}}_{k-1} - \vect{\hat{z}}^{*} \right\|_2 \nonumber \\
    &= \left\|\vect{x}_+ - \vect{\hat{z}}^{*} - \alpha  \vect{\bar{g}}^{\hat{z}}_{k-1} + \alpha \vect{\bar{g}}^{\bar{x}}_{k-1} \right\|_2\nonumber\\
    &\leq \eta \left\|\vect{\bar{x}}_{k-1} - \vect{\hat{z}}^{*} \right\|_2 + \alpha \left\|\vect{\bar{g}}^{\hat{z}}_{k-1} - \vect{\bar{g}}^{\bar{x}}_{k-1} \right\|_2\nonumber \\
    &\leq \eta \left\| \vect{\bar{x}}_{k-1} - \vect{\hat{z}}^{*} \right\|_2 + \alpha L \left\| \vect{\hat{z}}_{k-1} - \vect{\bar{x}}_{k-1} \right\|_2\nonumber\\
    &\leq \eta \left\| \vect{\bar{x}}_{k-1} - \vect{\hat{z}}^{*} \right\|_2 + \alpha L \left\|Y_{k-1}^{-1}\vect{\hat{x}}_{k-1} - \vect{\bar{x}}_{k-1}\right\|_2\nonumber\\
    &\leq \eta \left\| \vect{\bar{x}}_{k-1} - \vect{\hat{z}}^{*} \right\|_2 \nonumber\\
    &\quad + \alpha L \left\|Y_{k-1}^{-1}\left(\vect{\hat{x}}_{k-1}-Y_{\infty} \vect{\bar{x}}_{k-1}\right)\right\|_2\nonumber\\ &\quad+ \alpha L \left\|\left(Y_{k-1}^{-1} Y_{\infty}-I_{\bar{n}}\right) \vect{\bar{x}}_{k-1}\right\|_2\nonumber\\
    &\quad + \alpha L \tilde{y} \psi \gamma_{1}^{k-1}\left\|\vect{\bar{x}}_{k-1}\right\|_2,
\end{align}
where the second equality holds by letting $\vect{x}_+ = \vect{\bar{x}}_{k-1} - \alpha \vect{\bar{g}}^{\hat{x}}_{k-1}$, the first inequality holds by applying Lemma~\ref{lem:bubeck}, and the second inequality by making use of the Lipschitz-continuity assumption from A3. Then, the third inequality is obtained by substituting $\vect{\hat{z}}_{k-1}$ from \eqref{eq:z}, and finally, we get the last inequality by applying \cite[Lemma~8]{xi2017add}.

For the third element, $\left\|\vect{\hat{w}}_k-Y_{\infty} \vect{\bar{g}}^{\hat{z}}_k\right\|$, we substitute \eqref{eq:w}, and we get:
\begin{align}\label{eq:bound_step3}
    &\left\|\mathbf{\hat{w}}_{k}-Y_{\infty} \vect{\bar{g}}^{\hat{z}}_k\right\|= \left\|\Xi \mathbf{\hat{w}}_{k-1} + \nabla \vect{\hat{f}}(\vect{\hat{z}}_{k}) - \nabla \vect{\hat{f}}(\vect{\hat{z}}_{k-1}) -Y_{\infty} \vect{\bar{g}}^{\hat{z}}_k\right\|\nonumber\\
    &\leq \left\|\Xi \mathbf{\hat{w}}_{k-1} - Y_{\infty} \vect{\bar{g}}^{\hat{z}}_{k-1}\right\| \nonumber\\
    &\quad + \left\| \left(\nabla \vect{\hat{f}}(\vect{\hat{z}}_{k}) - \nabla \vect{\hat{f}}(\vect{\hat{z}}_{k-1})\right) - \left(Y_{\infty} \vect{\bar{g}}^{\hat{z}}_k - Y_{\infty} \vect{\bar{g}}^{\hat{z}}_{k-1}\right) \right\|\nonumber\\
    &\leq \sigma \left\|\mathbf{\hat{w}}_{k-1} - Y_{\infty} \mathbf{\bar{w}}_{k-1} \right\| \nonumber\\
    &\quad + \left\| \left(\nabla \vect{\hat{f}}(\vect{\hat{z}}_{k}) - \nabla \vect{\hat{f}}(\vect{\hat{z}}_{k-1})\right) - \left(Y_{\infty} \vect{\bar{g}}^{\hat{z}}_k - Y_{\infty} \vect{\bar{g}}^{\hat{z}}_{k-1}\right) \right\|\nonumber\\
    &\leq \sigma \left\|\mathbf{\hat{w}}_{k-1} - Y_{\infty} \mathbf{\bar{w}}_{k-1} \right\| \nonumber\\
    &\quad + \left\| \left(I_{\bar{n}} - \Xi_{\infty}\right) \left(\nabla \vect{\hat{f}}(\vect{\hat{z}}_{k}) - \nabla \vect{\hat{f}}(\vect{\hat{z}}_{k-1}\right) \right\|_2\nonumber\\ 
    &\leq \sigma \left\|\mathbf{\hat{w}}_{k-1} - Y_{\infty} \mathbf{\bar{w}}_{k-1} \right\| + d \epsilon L \left\| \vect{\hat{z}}_{k} - \vect{\hat{z}}_{k-1} \right\|_2
\end{align}
where the second inequality comes from Lemma~\ref{lem:dynamics_initial_conditions} and Lemma~\ref{lem:sigma}. Then, since $\frac{1}{\bar{n}} Y_{\infty} \mathbf{1}_{\bar{n}} \mathbf{1}_{\bar{n}}^{\top}=\Xi_{\infty}$, we get the third inequality, while by using the Lipschitz-continuity assumption A3, we obtain the last inequality.
Now we need to bound $\left\|\mathbf{\hat{z}}_{k}-\mathbf{\hat{z}}_{k-1}\right\|_2$:
\begin{align}
\left\|\vect{\hat{z}}_{k}-\vect{\hat{z}}_{k-1}\right\|_2 \leq &\left\|Y_{k}^{-1}\left(\vect{\hat{x}}_{k}-\vect{\hat{x}}_{k-1}\right)\right\|_2 \nonumber\\
&+\left\|\left(Y_{k}^{-1}-Y_{k-1}^{-1}\right) \vect{\hat{x}}_{k-1}\right\|_2 \nonumber\\
\leq &\left\|Y_{k}^{-1}\left(\Xi-I_{\bar{n}}\right) \vect{\hat{x}}_{k-1}\right\|_2+\alpha\left\|Y_{k}^{-1} \vect{\hat{w}}_{k-1}\right\|_2 \nonumber\\
&+\left\|Y_{k}^{-1}-Y_{k-1}^{-1}\right\|_2\left\|\vect{\hat{x}}_{k-1}\right\|_2 \nonumber\\
\leq &\left(\tilde{y} \xi+\alpha \tilde{y}^{2} y L\right)\left\|\vect{\hat{x}}_{k-1}-Y_{\infty} \vect{\bar{x}}_{k-1}\right\|_2 \nonumber\\
&+\alpha \tilde{y}\left\|\vect{w}_{k-1}-Y_{\infty} \vect{\bar{g}}^{\hat{z}}_{k-1}\right\|_2 \nonumber\\
&+\alpha \tilde{y} y L\left\|\vect{\bar{x}}_{k-1}-\vect{\hat{z}}^{*}\right\|_2 \nonumber\\
&+(\alpha y L+2) \tilde{y}^{2} \psi \gamma_{1}^{k-1}\left\|\vect{\hat{x}}_{k-1}\right\|_2,\label{eq:z-zprev}
\end{align}
where the first inequality is obtained by substituting \eqref{eq:z}, second inequality by substituting \eqref{eq:x}, third inequality holds due to $\left(\Xi-I_{\bar{n}}\right) Y_{\infty} \vect{\hat{x}}_{k-1}=\mathbf{0}_{\bar{n}}$, and:
\begin{align}
    \left\| \vect{\bar{g}}^{\bar{x}}_{k-1} \right\| = 
    \left\| \frac{1}{\bar{n}} \mathbf{1}_{\bar{n}} \mathbf{1}_{\bar{n}}^{\top} \nabla \vect{\hat{f}}(\vect{\bar{x}}_{k-1}) \right\|_2 \leq L\left\|\vect{\bar{x}}_{k-1}-\vect{\hat{z}}^{*}\right\|,
\end{align}
which bounds the following:
\begin{align}
\left\|Y_{k}^{-1} \mathbf{\hat{w}}_{k-1}\right\| \leq &\left\|Y_{k}^{-1}\left(\mathbf{\hat{w}}_{k-1}-Y_{\infty} \vect{\bar{g}}^{\hat{z}}_{k-1}\right)\right\|_2 \nonumber\\
&+\left\|Y_{k}^{-1} Y_{\infty} \vect{\bar{g}}^{\bar{x}}_{k-1}\right\|_2 \nonumber\\
&+\left\|Y_{k}^{-1} Y_{\infty}\left(\vect{\bar{g}}^{\hat{z}}_{k-1}-\vect{\bar{g}}^{\bar{x}}_{k-1}\right)\right\|_2 \nonumber\\
\leq & \tilde{y}\left\|\mathbf{\hat{w}}_{k-1}-Y_{\infty} \vect{\bar{g}}^{\hat{z}}_{k-1}\right\|_2 \nonumber\\
&+\tilde{y} y L\left\|\vect{\bar{x}}_{k-1}-\vect{\hat{z}}^{*}\right\|_2 \nonumber\\
&+\tilde{y} y L\left\|\vect{\hat{z}}_{k-1}-\vect{\bar{x}}_{k-1}\right\|_2 \nonumber\\
\leq & \tilde{y}\left\|\mathbf{\hat{w}}_{k-1}-Y_{\infty} \vect{\bar{g}}^{\hat{z}}_{k-1}\right\|_2 \nonumber\\
&+\tilde{y} y L\left\|\vect{\bar{x}}_{k-1}-\vect{\hat{z}}^{*}\right\|_2 \nonumber\\
&+\tilde{y}^{2} y L\left\|\vect{\hat{x}}_{k-1}-Y_{\infty} \vect{\bar{x}}_{k-1}\right\|_2 \nonumber\\
&+\tilde{y}^{2} y L \psi \gamma_{1}^{k-1}\left\|\vect{\hat{x}}_{k-1}\right\|_2,
\end{align}
where the last inequality holds due to the bound of $\left\| \vect{\hat{z}}_{k-1} - \vect{\bar{x}}_{k-1} \right\|$ obtained from the last inequality of \eqref{eq:bound_step2}.
Finally, by substituting \eqref{eq:z-zprev} in \eqref{eq:bound_step3} we get the bound on the third element, given by:
\begin{align}
\!\!\!\!  \left\|\vect{\hat{w}}_{k}-Y_{\infty} \vect{\bar{g}}^{\hat{z}}_k\right\| \leq& \left(c d \epsilon L \xi \tilde{y}+ \kappa \right)\left\|\vect{\hat{x}}_{k-1}-Y_{\infty} \vect{\bar{x}}_{k-1}\right\| \nonumber\\
&+\alpha d \epsilon L^{2} y \tilde{y}\left\|\vect{\bar{x}}_{k-1}-\vect{\hat{z}}^{*}\right\|_2\nonumber\\
&+\left(\sigma+\alpha c d \epsilon L \tilde{y}\right)\left\|\vect{\hat{w}}_{k-1}-Y_{\infty} \vect{\bar{g}}^{\hat{z}}_{k-1}\right\| \nonumber\\
&+(\alpha y L+2) d \epsilon L \tilde{y}^{2} \psi \gamma_{1}^{k-1}\left\|\vect{\hat{x}}_{k-1}\right\|_2,
\end{align}
where $\kappa=\alpha c d \epsilon L^{2} y \tilde{y}^{2}$.

Proof completed with all elements bounded.
\end{proof}

\begin{proof}[\textbf{Proof of Lemma~\ref{lem:step-size-range}}]
To derive the range of the gradient step-size $\alpha$ in \eqref{eq:alpha} that maintains the spectral radius $\rho(G_a)<1$, we adopt similar analysis as in \cite{xi2017add} \emph{mutatis mutandis}, by replacing $n$, with the total number of nodes in the augmented digraph $\bar{n}$.
\end{proof}

\begin{lem}[Auxiliary relation]\label{lem:dynamics_initial_conditions}
Consider the variables $\vect{\bar{x}}_{k}$ and $\vect{\bar{w}}_{k}$ from \eqref{eq:generic_notation} that incorporate the transmission delays due to the exchange of information between the nodes. The following equations hold for all $k$ :
$$ \text{(a) } \vect{\bar{w}}_{k}= \vect{\bar{g}}^{\hat{z}}_k, \text{ and (b) } \vect{\bar{x}}_{k+1}=\vect{\bar{x}}_{k}- \alpha \vect{\bar{g}}^{\hat{z}}_k.$$

\begin{proof} 
From the column-stochasticity property of $\Xi$, and \eqref{eq:w} we have that
$
\vect{\bar{w}}_{k} 
= \vect{\bar{w}}_{k-1} + \vect{\bar{g}}^{\hat{z}}_k - \nabla \vect{\bar{g}}^{\hat{z}}_{k-1}.\nonumber
$
Writing $\vect{\bar{w}}_k$ recursively we get $\vect{\bar{w}}_{k}=\vect{\bar{w}}_{0}+\vect{\bar{g}}^{\hat{z}}_k-\vect{\bar{g}}^{\hat{z}}_0$, which gives $\vect{\bar{w}}_k=\vect{\bar{g}}^{\hat{z}}_k$ due to the initial conditions in \eqref{eq:w}, and $\vect{\hat{w}}_0=\nabla \vect{\hat{f}}(\vect{\hat{z}}_0)$. From \eqref{eq:x} we have that
$
\vect{\bar{x}}_{k+1} = \frac{1}{\bar{n}} \mathbf{1}_{\bar{n}} \mathbf{1}_{\bar{n}}^{\top} \left(\Xi \hat{\vect{x}}_{k}-\alpha \vect{\hat{w}}_{k}\right)=
\vect{\bar{x}}_{k}-\alpha \vect{\bar{g}}^{\hat{z}}_k.\nonumber
$\end{proof}
\end{lem}

\begin{lem}[Bubeck \cite{bubeck2014convex}]\label{lem:bubeck}
Let Assumption A3 hold for the objective functions $f_{i}(\mathbf{x})$ in \eqref{eq:optimization_problem}. For any $\mathbf{x} \in \mathbb{R}^{p}$ define $\mathbf{x}_{+}=\mathbf{x}-\alpha \nabla \mathbf{f}(\mathbf{x})$, where $0<$ $\alpha<\frac{2}{\bar{n} L}$. Then
\begin{align}
\left\|\mathbf{x}_{+}-\underline{\mathbf{x}}^{*}\right\| \leq \eta\left\|\mathbf{x}-\underline{\mathbf{x}}^{*}\right\|
\end{align}
where $\eta=\max (|1-\alpha \bar{n} L|,|1-\alpha \bar{n} \mu|)$.
\end{lem}

\bibliographystyle{unsrt} 
\bibliography{references}

\end{document}